\documentclass[10pt,conference]{IEEEtran}
\IEEEoverridecommandlockouts
\usepackage{cite}
\usepackage{amsmath,amssymb,amsfonts}
\usepackage{algorithm}
\usepackage{algpseudocode}
\usepackage{graphicx}
\usepackage{textcomp}
\usepackage{xcolor}
\usepackage{autobreak}
\usepackage[english]{babel}
\usepackage{amsthm}

\newtheorem{theorem}{Theorem}[section]

\def\BibTeX{{\rm B\kern-.05em{\sc i\kern-.025em b}\kern-.08em
    T\kern-.1667em\lower.7ex\hbox{E}\kern-.125emX}}
\begin{document}

\title{Privacy-Preserving Quantum Annealing for Quadratic Unconstrained Binary Optimization (QUBO) Problems}

\author{
\IEEEauthorblockN{ Moyang Xie$^1$, Yuan Zhang$^1$, Sheng Zhong$^1$, Qun Li$^2$ }
\IEEEauthorblockA{
$^1$\textit{State Key Laboratory for Novel Software Technology, Nanjing University,} Nanjing, China \\
$^2$\textit{Department of Computer Science, William \& Mary,} Williamsburg, VA, USA \\
xie\_moyang@foxmail.com, zhangyuan@nju.edu.cn, zhongsheng@nju.edu.cn, liqun@cs.wm.edu
}

}


\maketitle

\begin{abstract}

Quantum annealers offer a promising approach to solve Quadratic Unconstrained Binary Optimization (QUBO) problems, which have a wide range of applications. However, when a user submits its QUBO problem to a third-party quantum annealer, the problem itself may disclose the user's private information to the quantum annealing service provider. To mitigate this risk, we introduce a privacy-preserving QUBO framework and propose a novel solution method. Our approach employs a combination of digit-wise splitting and matrix permutation to obfuscate the QUBO problem's model matrix $Q$, effectively concealing the matrix elements. In addition, based on the solution to the obfuscated version of the QUBO problem, we can reconstruct the solution to the original problem with high accuracy. Theoretical analysis and empirical tests confirm the efficacy and efficiency of our proposed technique, demonstrating its potential for preserving user privacy in quantum annealing services.

\end{abstract}


\section{Introduction}\label{sec:int}

Leveraging the unique quantum mechanics, such as superposition, entanglement, and quantum interference, quantum computers provide a new computing paradigm for many new applications. 
For example, quantum annealers, based on a process of simulated annealing, can search the solution space efficiently and find the optimal solution to optimization problems~\cite{das2008colloquium} that are difficult for classical computers to solve. 

In particular, quantum annealers are known for effectively solving the Quadratic Unconstrained Binary Optimization (QUBO)~\cite{battaglia2005optimization,rajak2023quantum}, which has broad-ranging applications in fields such as finance, economics, etc.~\cite{kochenberger2014unconstrained}. 
As an NP-hard problem, QUBO poses a significant computational challenge to classical computers. It has been applied to various classical problems in theoretical computer science, including the maximum cut, graph coloring, and the partition problem, all of which can be formulated as QUBO embeddings~\cite{glover2019quantum}. Additionally, QUBO has extended its utility to machine learning, with embeddings designed for support-vector machines, clustering algorithms, and probabilistic graphical models~\cite{mucke2019learning,date2021qubo}. The close relationship between QUBO and Ising models positions QUBO as a central problem class in adiabatic quantum computation, where solutions are sought through the physical process known as quantum annealing. A diverse array of industry leaders, including D-Wave, Google, IBM, Microsoft, and Amazon, alongside esteemed public institutions such as Oak Ridge National Laboratory, Los Alamos National Laboratory, and NASA's Ames Research Center, have dedicated substantial resources to exploring the utilization of quantum computers for resolving QUBO problems~\cite{glover2019quantum}. This extensive body of research underscores the immense value and untapped potential of leveraging quantum computational resources for QUBO problem-solving.

Unfortunately, to date, it is hard for ordinary individuals or even companies to have their own quantum computer. Noticing this unmet demand, companies such as D-Wave have started to provide cloud services of quantum annealing. Although this cloud-based service is convenient for users to solve their QUBO problems remotely, it also brings immediate privacy concerns. Specifically, the QUBO instances stemming from users' real-world challenges may encapsulate sensitive information, ranging from financial data and corporate secrets to personal health records. Consequently, when users upload their QUBO problems, there exists a risk of privacy infringement by the service provider, who could potentially scrutinize these problems. Hence, the development of a robust method to safeguard user privacy while enabling utilization of third-party quantum annealing services for solving QUBO problems is paramount.

This paper addresses the privacy-preserving QUBO solving problem by introducing a novel approach. Our method generates an obfuscated version of the original QUBO problem, which, when solved, yields a solution that closely approximates the original problem's solution with high accuracy. Instead of submitting their original problem to the quantum annealing service provider, users can send this obfuscated version. This strategy allows users to leverage the computational power of quantum annealing for obtaining high-quality solutions to their QUBO problems while significantly reducing the risk of privacy breaches. Our approach strikes a balance between utilizing advanced quantum annealing capabilities and safeguarding sensitive information, potentially paving the way for more widespread and secure adoption of cloud-based quantum annealing services.

Our major contributions can be summarized as follows.
\begin{itemize}
    \item We formulate the privacy-preserving QUBO problem and propose an obfuscation solution.
    \item We provide theoretical analyses on the security and efficiency of our method. 
    \item Additionally, we conduct experiments to evaluate the efficacy and efficiency of our proposed technique. Results demonstrate that our method achieves good accuracy.
    \end{itemize}

The rest of this paper is organized as follows. In Sec.~\ref{sec:ppQUBO}, we formulate the privacy-preserving QUBO problem.
In Sec.~\ref{sec:our_app}, we explain our method in details. Theoretical analyses and experimental evaluations are presented in Sec.~\ref{sec:per_ana} and Sec.~\ref{sec:eva_acc} respectively. After discussing related work in Sec.~\ref{sec:rel_wor}, we conclude our paper in Sec.~\ref{sec:con_fut}.

\section{The Privacy-preserving QUBO Solving Problem}\label{sec:ppQUBO}

The Quadratic unconstrained binary optimization (QUBO) is a model for solving combinatorial optimization problems that can be applied to a variety of important and difficult combinatorial problems such as Multiple Knapsack Problems, P-Median Problems, Maximum Independent Set Problems, Maximum Cut Problems, Graph Coloring Problems, SAT problems and so on. The problem has the following form:

\begin{center}
minimize $y=x^tQx$,
\end{center}
where $x\in\{0,1\}^n$ is an vector of binary elements, and $Q\in \mathbb{R}^{n\times n}$ is the ``model matrix'' which is a square matrix of constants that determines the objective function to be optimized or minimized over all possible values of $x$ in our case.

For example, one might aim to minimize a function $y=f(x_1,x_2,x_3,x_4)$ as follows.

\medskip
$
\begin{aligned}
y & =6x_1^2-18x_3^2-2x_4^2+10x_1x_3-18x_1x_4-6x_2x_3+\\& \ \ \ 4x_2x_4+4x_3x_4\\& = \left(\begin{array}{l}x_{1} \\ x_{2} \\ x_{3} \\ x_{4}\end{array}\right)^T\left(\begin{array}{cccc}6 & 0 & 5 & -9 \\ 0 & 0 & -3 & 2 \\ 5 & -3 & -18 & 2 \\ -9 & 2 & 2 & -2\end{array}\right)\left(\begin{array}{l}x_{1} \\ x_{2} \\ x_{3} \\ x_{4}\end{array}\right).
\end{aligned}
$
\medskip

Accordingly, the above optimization problem can be determined by a square matrix:

\medskip
$
\begin{aligned}
Q & =\left(\begin{array}{cccc}6 & 0 & 5 & -9 \\ 0 & 0 & -3 & 2 \\ 5 & -3 & -18 & 2 \\ -9 & 2 & 2 & -2\end{array}\right),
\end{aligned}
$
\medskip

\noindent and a brute-force search can find the solution to the problem: $x_1=0,\ x_2=1,\ x_3=1,\ x_4=0$.

In this paper, we consider the \textit{privacy-preserving QUBO solving problem} as follows.


A quantum annealing \textit{server} solves QUBO problems. The \textit{client} submits a QUBO model matrix and receives the optimal solution. The client aims to solve its QUBO problem with matrix $Q$ while keeping $Q$ private from the server. To achieve this, the client sends a series of matrices $\{M\}$ to the server and receives optimal solutions $\{x'\}$, which are used to derive the optimal solution $x$ of $Q$. Our goals are: (1) prevent the server from inferring $Q$ from $\{M\}$ for privacy preservation, and (2) accurately derive $x$ from $\{x'\}$ to solve the client's problem.



\section{Our approach}\label{sec:our_app}


Our approach generates an obfuscated version of the model matrix. We split the QUBO model matrix $Q$ into matrices whose weighted sum equals $Q$. We represent each element of $Q$ using a radix $r$ ($r\in \mathbb{N}$) and split it into a string of $r$-radix digits, decomposing $Q$ into ``digit matrices". Each digit matrix contains elements' digits from the same digit place (e.g., the 1st digit's place, the 2nd digit's place, $\ldots$). 

We randomly permute rows and columns of these matrices and send them to the server in random order. The server cannot recover the original matrix without knowing the permutations or sending order.

Our approach preserves the relationship between the original matrix and sent matrices. The client can reconstruct the original matrix by summing the sent matrices with correct permutations and weights. After the server returns the optimization solutions, which are a series of binary vectors, corresponding to the matrices it received from the client, the client can recover the binary vectors in the correct order, and combine them to obtain a solution to its original QUBO problem.

In the following, we describe in details: 1) how the client constructs obfuscated matrices to send; and 2) how the client reconstructs the final optimal solution $x$ from returned binary vectors.

\subsection{Obtaining the Matrices to Be Sent}

Following the operations below, the client computes $k$ matrices $M^\sigma _1$, $M^\sigma _2$,..., $M^\sigma _k$ $\in \mathbb{R}^{n\times n}$ from the model matrix $Q$, and sends them to the server in a random order, which can be completed in one single message transmission.

{\textit{1) Matrix normalization.} To enhance both privacy and efficiency in our approach, we normalize the model matrix $Q$ while preserving the relative ratio between its elements. This normalization process scales all elements to fall within the interval (-1, 1), with the element having the largest absolute value approaching either -1 or 1. Consequently, the integer parts of all elements become 0, simplifying our subsequent digit-wise splitting operation to only consider digits after the decimal point. By ensuring the element with the largest absolute value is near -1 or 1, we try to maximize the number of elements with a non-zero first decimal digit. This first digit is crucial as it contributes significantly to the overall value of each element. As a result, the matrix formed by these first digits holds the highest weight among all matrices derived from the radix decomposition. Its corresponding optimal binary vector also tends to carry the greatest weight in the subsequent weighted summation process. Maximizing non-zero first digits enhances the significance of both the corresponding matrix and its optimal binary vector, thereby improving result accuracy. Moreover, this normalization step helps to obscure the actual element values within the model matrix $Q$, thereby enhancing the security of our approach. Specifically, the client computes the normalized matrix $Q^*$ as:

\begin{equation}
    Q^*[i,j]:=\frac{Q[i,j]}{(1+\epsilon)max(Q)},\ \forall i,j \in \left\{ 1, ... , n \right\},   
\end{equation}
where $\max(Q)$ denotes the maximum absolute value of all elements in $Q$, $Q[i,j]$ and $Q^*[i,j]$ denote the $i$th-row-$j$th-column element of $Q$ and $Q^*$ resp., and $\epsilon \rightarrow 0^+$ is a small positive real number that is close to 0.

\textit{2) Matrix decomposition with base $r$.} The client represents (the non-integer part of) normalized model matrix $Q^*$ using radix $r$, takes all elements' most significant digits to form $M_1$, the second most significant digits as $M_2$, ... , and the $k$th digit as $M_k$:

\begin{equation}
    M_m[i,j]:=\text{the}\ m \text{th}\ \text{digit} \ \text{of} \ Q^*[i,j],
\end{equation}

\begin{center}
    $\forall m\in \left\{1,2, ... ,k \right\}, \forall i,j \in \left\{1,2, ... ,n \right\}$
\end{center}

\textit{3) Random matrix permutation.} The client generates $k$ random permutations $\sigma_1$, $\sigma_2$, ... , $\sigma_k$ and uses them to permute the rows and columns of $M_1$, $M_2$, ... , $M_k$ respectively to obtain $M_1^\sigma$, $M_2^\sigma$, ... , $M_k^\sigma$ as:

\begin{equation}
    M_m^\sigma[i,j]:=M_m[\sigma_m(i),\sigma_m(j)] 
\end{equation}

\begin{center}
    $\forall m \in \left\{1,2, ... ,k \right\}, \forall i,j \in \left\{1,2, ... ,n \right\}$
\end{center}

\textit{4) Sending in a random order.} To perform this, the client generates a random permutation $\sigma'$ and sends $M_{\sigma'(1)}^\sigma$, $M_{\sigma'(2)}^\sigma$, ... , $M_{\sigma'(k)}^\sigma$ in order.\newline


\noindent\textbf{Note:} To further obfuscate the matrix $Q$, the client can also generate and send additional random matrices to the server along with the decomposed and permuted matrices derived from $Q$. These random matrices serve as decoys, making it more difficult for the server to identify or extract sensitive information about the original matrix $Q$. Upon receiving the solutions from the server corresponding to all the matrices, including the random ones, the client can simply discard the solutions related to the random matrices and focus only on the solutions pertaining to the decomposed and permuted matrices of $Q$. This extra layer of obfuscation further enhances the privacy protection of the sensitive data within the matrix $Q$. For presentation simplicity, we omit this step in our approach description.

\subsection{Reconstruction of the optimal solution}

After receiving $k$ obfuscated matrices from the client, the server sends back the corresponding $k$ optimization solutions or binary vectors to the client. We denote these binary vectors as $v_1^{\sigma'} $, $v_2^{\sigma'} $, ... , $v_k^{\sigma'} $, which are the solutions to the binary optimization for $M_{\sigma'(1)}^\sigma$, $M_{\sigma'(2)}^\sigma$, ... , $M_{\sigma'(k)}^\sigma$. The client uses inverse permutations of $\sigma_1,\sigma_2,\ldots, \sigma_k$ and $\sigma'$ to recover them from the server's response, and then sum them with corresponding weights. Finally, the client computes a probability distribution from the weighted sum vector, and samples the final optimal $x$ over the distribution. Specifically, the above operations are completed by executing the following steps:

\textit{1) Inversing permutations.} The client first permutes $v_1^{\sigma'} $, $v_2^{\sigma'} $, ... , $v_k^{\sigma'} $ using $\sigma^{'-1}$ (i.e. the inverse of $\sigma^{'}$) to get the correctly ordered $v_1^{\sigma} $, $v_2^{\sigma} $, ... , $v_k^{\sigma}$ which are binary optimization solutions to $M_1^\sigma$, $M_2^\sigma$,... , $M_k^\sigma$:

\begin{equation}
    v_m^\sigma:=v_{\sigma^{'-1}(m)}^{\sigma'} ,\ \forall m \in \left\{1,2, ...,k \right\} 
\end{equation}

Then the client permutes the elements of $v_1^{\sigma} $, $v_2^{\sigma} $, ... , $v_k^{\sigma}$ using $\sigma^{-1}_1$, $\sigma^{-1}_2$, ... , $\sigma^{-1}_k$ respectively to recover $v_1$, $v_2$, ... , $v_k$ which are binary optimization solutions to $M_1$, $M_2$, ... , $M_k$: 
\begin{equation}
    v_m[i]:=v_m^{\sigma}[\sigma^{-1}_m(i)] ,\ \forall m \in \left\{1,2, ... ,k \right\},\ \forall i\in \left\{1,2, ... ,n \right\}
\end{equation}


\textit{2) Weighted average calculation.}  The client sets a weight vector $w\in\mathbb{R}^k$ and computes the weighted average $\tilde{x}$ of $v_1$, $v_2$, ... , $v_k$ with weights of $w[1]$, $w[2]$, ... , $w[k]$ as: 

\begin{equation}
    \tilde{x}:=\frac{\sum^{k}_{m=1}{w[m]\cdot v_m}}{\sum^k_{m=1}{w[m]}}.
\end{equation}

\noindent$\tilde{x}$ is used as an estimation on the probability that each entry in the optimal $x$ equals 1. 

Intuitively, the matrices with smaller subscripts in $M_1$, $M_2$, ... , $M_k$ are split from higher digits in the model matrix $Q$ and contribute more to $x^tQx$, and the corresponding optimization solutions should also have larger weights in $x$. 
A possible $w$ is
\begin{equation}
    w[m]:=\frac{1}{r^m},\ \forall m \in \left\{1,2,... ,k \right\}.
\end{equation}
We will discuss the choice of $w$ again later in the accuracy analysis.

\textit{3) Sampling and selection.} Finally, the client independently samples $t$ vectors $X_1$, $X_2$, ... , $X_t$ from the distribution: 

\begin{equation}
    X \stackrel{}{\leftarrow}[0,1]^n,Pr[X[i]=1]=\tilde{x}[i],
\end{equation}

 and selects the one that minimizes the objective function 
\begin{equation}
     \mathop{argmin}\limits_{x\in \left\{X_1,X_2,...,X_t \right\}} \ x^tQx
\end{equation}
 as the final optimal $x$.

\section{Performance Analysis}\label{sec:per_ana}

Our approach has the following principal parameters:
\begin{itemize}
\item $n$ - the order of the matrix, determined by the model;
\item $r$ - the radix or base in our numeral system, determined by the client;
\item $k$ - the total number of ``digit matrices'', determined by the client;
\item $t$ - the total number of random samplings on $\tilde{x}$, determined by the client.
\end{itemize}
In this section, we discuss the privacy protection, accuracy and efficiency of our approach and analyze the effects of the above parameters.

\subsection{Privacy protection}

First of all, due to the setup of the normalization step, $Q^*$ hides the element-specific values of $Q$ and only retains the proportionality between them, so neither the element-specific values of $Q$ nor the optimal $x^TQx$ values will be leaked. Therefore, there can only be leakage of the proportionality of the elements in $Q$, and we analyze the leakage of this below.

An optimistic, but without loss of generality, assumption is that the distribution of the digits of each element of $Q^*$ in each $r$-base position is indistinguishable from a uniformly independent distribution, such that the absolute values of the elements in $M_1$, $M_2$, ... , $M_k$ are also uniformly independently distributed. In this case, the server can only estimate a permutation by indices of the sign matrix of $Q$ based on the sign matrices of any of the received matrices. It is discussed next whether the server can get more valuable information from the received matrices $M_{\sigma'(1)}^\sigma$, $M_{\sigma'(2)}^\sigma$, ... , $M _{\sigma'(k)}^\sigma$.

\begin{theorem}
If the absolute values of all digit matrices' elements follow independently uniform distribution over the digit value space,  
the probability of recovering a correct $Q^*$ is  $\frac{1}{\alpha^{k-1}k!n!}$, where $\alpha$ denotes the number of automorphisms of the sign matrix of $Q$ under the row/column index permutation transformation.
\end{theorem}

\begin{proof}
Note that before the random permutations on rows and columns, all $k$ digital matrices $\{M_m\}_{m=1,\ldots, k}$ share the same sign matrix (i.e. the matrix formed by all elements' signs) with the $Q$ or $Q^*$. Therefore the correct digital matrices can be recovered only by a combination of permutations $\{\sigma^*_m\}_{m=1,\ldots,k}$ such that applying them on the permuted matrices $\{ M^\sigma_m\}_{m=1,\ldots,k}$ respectively would make their sign matrices become identical. Let $\Sigma$ be the set of all such combinations.

Since the absolute values of the elements in digital matrices are uniformly and independently distributed, the server is not able to distinguish between the correct combination of permutations (i.e. $\{\sigma^{-1}_m\}_{m=1,\ldots,k}$ ) that recovers the true digital matrices and other incorrect combination of permutations in $\Sigma$. 

For any $\sigma^*_1$, there are $\alpha ^{k-1}$ possible $(\sigma^*_2, ... ,\sigma^*_k)$ that the sign matrices of the resulting matrices are the same after permutation. 

Therefore, we know the size of $\Sigma$ equals $\alpha^{k-1}n!$, and the probability that the server correctly guesses $\{\sigma^{-1}_m\}_{m=1,\ldots,k}$ and recovers $\{M_{\delta'(m)}\}$ that have correct rows and columns equals $\frac{1}{\alpha^{k-1}n!}$. 

Finally, there are $k!$ possible orders that the digit matrices are sent with. Therefore, the probability that the server recovers the correct $Q^*$ equals 
\begin{center}
    $\frac{1}{\alpha^{k-1}k!n!}$.
\end{center}






\end{proof}

In the worst-case scenario, where the absolute values of the elements in \( M_1 \), \( M_2 \), and \( M_k \) are not uniformly and independently distributed, we can employ the optional privacy-enhancing method mentioned earlier. By sending the server specific confusion matrices, we ensure that the matrices appear to match a uniformly independent distribution or another nonsensical distribution. This approach is effective but requires contextual analysis.

When a large \( k \) is not necessary, we can achieve the privacy preservation effect of a larger \( k \) by sending more confusion matrices in this manner.

From the analysis above, we find that larger values of \( n \) and \( k \) improve privacy performance. With a fixed \( n \), a larger \( k \) means the server receives more matrices related to the model, which intuitively provides more information. However, paradoxically, it becomes more challenging to recover the private data. This suggests that increasing the amount of information increases the difficulty for the server to organize and analyze it, due to the larger probability space created by more matrices.

Although the parameter \( r \) does not appear in the equation, its impact on privacy is significant. As discussed, ``recovering the first \( k' (k' < k) \) correct digits of \( Q^* \)" shows that a larger \( r \) provides more space for the first \( k' \) digits, thus containing more information. In an extreme case, where \( r \) is very large, \( M_1 \) alone may contain almost all the information in \( Q \). Conversely, a smaller \( r \) makes the distribution of each element look more uniform and independent, thereby enhancing privacy performance. Therefore, smaller values of \( r \) improve privacy protection.

\subsection{Accuracy} 
First, we make several qualitative observations regarding the impact of various parameters. A larger 
$t$ results in a smaller error because a greater number of samples increases the likelihood of obtaining a better binary vector. The parameter $k$ influences the error up to a point and then has a marginal effect. A too-small number of matrices primarily contain higher digits, potentially lacking sufficient information, whereas a too-large number of matrices include many lower digits, which contribute minimally to the results and appear more random. Increasing 
$r$ reduces the error, as a larger base encapsulates more information in each digit. In an extreme case, if $r$ is sufficiently large, sending only the first matrix to the server can yield an adequately accurate result.

We have conducted numerous experiments demonstrating that our approach achieves excellent accuracy (see Section~\ref{sec:eva_acc} for details). This can be intuitively understood: higher digits contribute more significantly to the model matrix \( Q \), thus having a greater impact on \( x^tQx \) and the optimal \( x \). The influence of a higher digit on the optimal \( x \) is reflected in the optimal solution of the matrix formed by that digit. The weight of this effect on the optimal \( x \) is captured in the final result through the weights we assign to that binary vector in the weighted average. Each lower digit reduces the weight in \( Q \) by a factor of \( r \), leading us to propose that its weight in the final result should also be reduced by a factor of \( r \). This is the basis for the generic weight vector \( w[i]:=\frac{1}{r^i},\ \forall i \in \left\{1,2,...,k \right\} \).

However, this weight vector \( w \) is not universally applicable. The selection of \( w \) should be related to the distribution of the elements in \( Q \). For example, when the absolute values of the non-zero elements in \( Q \) are similar in magnitude, focusing only on the highest position may suffice, and considering lower positions might reduce accuracy, especially when \( r \) is small. When the magnitude distribution of elements in \( Q \) is uncertain, determining an optimal \( w \) is challenging. Conversely, if the magnitude distribution is known, techniques such as multiple linear regression or deep learning can be employed to learn a more appropriate \( w \).

Given the complexity and uncertainty of this problem, providing a rigorous theoretical analysis of the accuracy is difficult. A more in-depth analysis of the choice of \( w \) and its impact on accuracy will be the focus of our future research.

\subsection{Efficiency}

\begin{table*}[htbp]
\caption{Computation, communication and storage costs}
\begin{center}
	\begin{tabular}{cccccc}
     \hline
		\multicolumn{1}{c}{\textbf{Communication}} & \multicolumn{1}{c}{\textbf{Rounds}} &\multicolumn{1}{c}{\textbf{Client Work}}  &  \multicolumn{1}{c}{\textbf{Server Work}} & \multicolumn{1}{c}{\textbf{Client Storage}}& \multicolumn{1}{c}{\textbf{Server Storage}}\\
		\multicolumn{1}{c}{$kn^2$} & \multicolumn{1}{c}{1} & \multicolumn{1}{c}{$kn^2+tn^2$} & \multicolumn{1}{c}{$k \ ^*$} &\multicolumn{1}{c}{$kn^2+tn$}&\multicolumn{1}{c}{$kn^2$}\\
    \hline
    \multicolumn{4}{l}{$^*$Here we count the overhead of the server solving a QUBO as $\mathcal{O}(1)$.}
	\end{tabular}
\end{center} 
\end{table*}

Our method demonstrates strong performance in terms of efficiency. For \textit{communication}, only one round is required: the client sends \( k \) \( n \)-order matrices to the server, and the server returns \( k \) \( n \)-dimensional binary vectors, resulting in a communication overhead of \(\mathcal{O}(kn^2)\). For \textit{computation}, the client needs to generate and permute \( k \) \( n \)-order matrices, which is \(\mathcal{O}(kn^2)\). After receiving the \( k \) \( n \)-dimensional binary vectors, the client computes their weighted average, which is \(\mathcal{O}(kn)\). Following this, the client samples \( t \) candidate vectors and computes the values of the corresponding quadratic forms, amounting to \(\mathcal{O}(tn^2)\). The sorting process to obtain the final result is \(\mathcal{O}(t \log t)\), which is typically subsumed under \(\mathcal{O}(tn^2)\). Thus, the total computational overhead for the client is \(\mathcal{O}(kn^2 + tn^2)\). The server, on the other hand, needs to solve QUBO for \( k \) \( n \)-order model matrices, with a total overhead of \(\mathcal{O}(k)\) assuming each solution takes \(\mathcal{O}(1)\) time.

For \textit{storage}, the client must store \(\mathcal{O}(k)\) \( n \)-order matrices and \(\mathcal{O}(t)\) \( n \)-dimensional vectors along with their corresponding quadratic values, leading to a total storage overhead of \(\mathcal{O}(kn^2 + tn)\). The server needs to store \( k \) \( n \)-order matrices and the corresponding \( k \) \( n \)-dimensional vectors, resulting in a total storage overhead of \(\mathcal{O}(kn^2)\).

As the number of transmitted matrices increases, the communication overhead increases linearly. Conversely, the difficulty for the server to learn information about the matrix \( Q \) increases exponentially.

\section{Evaluation on Accuracy}\label{sec:eva_acc}

To evaluate whether our approach achieves desirable accuracy, we conducted several experiments.

We randomly generated model matrices for testing under a normal distribution with a mean of 0 and a standard deviation of 4. Our observations indicate that optimizing these matrices is quite challenging, making them a rigorous test case. Given the current limitations in accessing a quantum annealer, we utilized the well-known commercial optimization solver IBM ILOG CPLEX to solve QUBO problems in our experiments. CPLEX is capable of providing optimal results at reasonably large problem scales~\cite{websiteCPLEX}.

We used the values of the optimization function \( x^tQx \) as the observed values and defined the accuracy of our solution as:
\[
   \text{acc} = \frac{\text{obtained value of } x^tQx}{\text{true optimal value of } x^tQx},
\]
and the error as:
\[
    \text{err} = \left| \frac{\text{obtained value of } x^tQx - \text{true optimal value of } x^tQx}{\text{true optimal value of } x^tQx} \right|.
\]
We focus on the difference between the obtained and true optimal values of the optimization function, as this difference is more critical than the difference between the obtained and true optimal solutions.


\begin{figure}[htbp]
\centerline{\includegraphics[width=0.8\columnwidth]{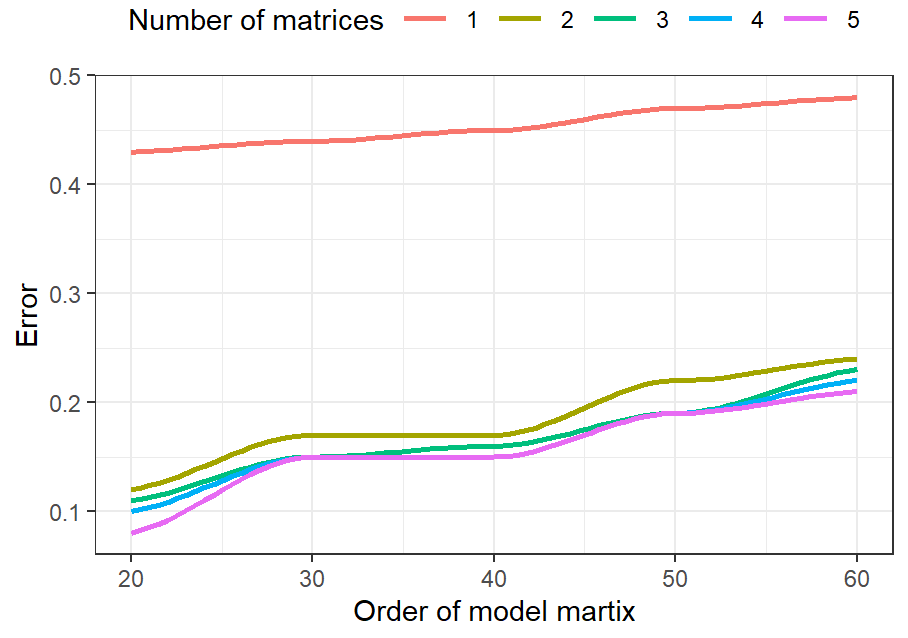}}
\caption{Error with different matrix orders and total numbers of matrices when sampling 200 times with a base of 2.}
\label{fig1}
\end{figure}

In Fig.~\ref{fig1}, we illustrate how the error varies as the matrix order increases. The results indicate that the error grows slowly with increasing order, following a trend that appears to be no more than linear, suggesting \(\mathcal{O}(n)\)-like behavior. These encouraging results imply that we can achieve sufficiently accurate solutions by adjusting other parameters, with an acceptable overhead. Fig.~\ref{fig1} also demonstrates the effect of increasing \( k \) on accuracy improvement, consistent with theoretical predictions. However, beyond a certain point, further increases in \( k \) offer diminishing returns in accuracy. Therefore, we can select a modest \( k \) and employ the transmission of confusion matrices to enhance security without significantly compromising accuracy.


\begin{figure}[htbp]
\centerline{\includegraphics[width=0.8\columnwidth]{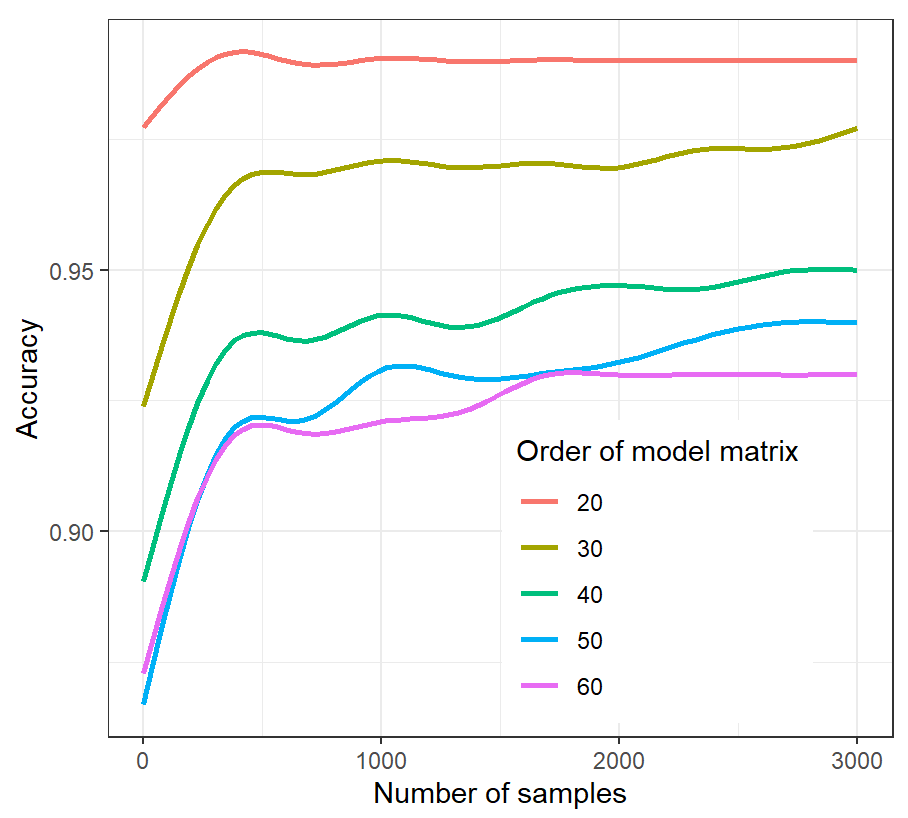}}
\caption{Accuracy with different total numbers of samples and orders of model matrix when the number of matrices is 5 with a base of 4.}
\label{fig2}
\end{figure}

Fig.~\ref{fig2} demonstrates the relationship between accuracy and the number of samples. As depicted, accuracy improves with an increasing number of samples. For each order of test matrices utilized, with 5 matrices and a base of 4, we achieved accuracy exceeding 90\% within 300 samples—a notably low overhead given the complexity of the problem. However, it's also evident that beyond a certain number of samples, the rate of accuracy improvement diminishes, indicating that achieving very high accuracy solely by increasing the number of random samples may still be challenging. The current method of random sampling is suboptimal, and we plan to explore more efficient techniques to derive the final \( x \) from the \([0,1]^n\) vector \(\tilde{x}\) at a reduced cost in the future.


\begin{figure}[htbp]
\centerline{\includegraphics[width=0.8\columnwidth]{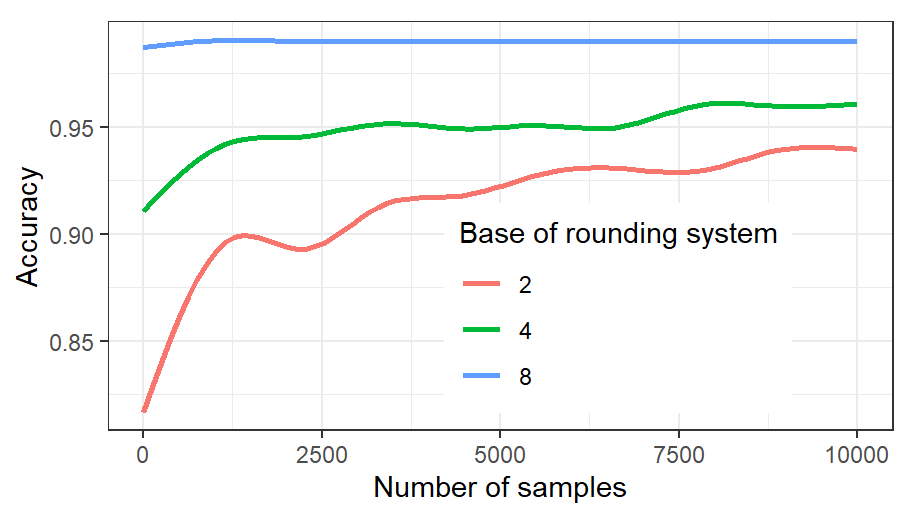}}
\caption{Accuracy with different total numbers of samples and bases of numeral/rounding system when the number of matrices is 5 and the order of model matrix is 40.}
\label{fig3}
\end{figure}

Fig.~\ref{fig3} illustrates the increase in accuracy as the base, \( r \), grows. The base is a flexible and influential parameter, allowing for significant accuracy improvements with its increase. However, this increase also comes at a cost to security, highlighting the trade-off between security, privacy, and overhead.


\section{Related Work}\label{sec:rel_wor}


Many quantum computing research projects focus on developing general quantum computing platforms for new applications, such as quantum machine learning~\cite{wu2024distributed, wu2022scalable, wu2023more}. This paper, however, examines quantum annealers, which are specialized quantum computers.

A considerable amount of research has been conducted on privacy-preserving methods for classical optimization. Notable works include~\cite{huang2015differentially,han2016differentially,shoukry2016privacy,8262869,7347330,8430960,zhang2022quadratic,han2009privacy,toft2009solving,Salinas2016}. However, these methods are not well-suited for our specific scenario, as quantum annealers differ fundamentally from general-purpose computation and continuous optimization, which these methods are designed for.


Specifically, the gradient descent method are considered in~\cite{han2009privacy,8262869,shoukry2016privacy,zhang2022quadratic,8430960}, while the simplex algorithm is studied in~\cite{toft2009solving}. However, the methods proposed in these works are designed for continuous or linear optimization and are not applicable to discrete quadratic optimization. Additionally, the quantum annealer does not support general-purpose computation, making it impossible to adopt these methods.

The works~\cite{7347330,Salinas2016} achieve privacy preservation through matrix transformations; however, their methods are only applicable to continuous optimization problems, while we are dealing with a discrete optimization problem.


A few works~\cite{huang2015differentially,han2016differentially} adopt differential privacy mechanisms~\cite{dwork2006differential} to protect individual participant's model matrix privacy in a distributed joint optimization task where the objective function to be minimized corresponds to the sum of all participants' model matrices. Although differential privacy can protect individual privacy, it offers little protection to the overall model matrix. Therefore, it cannot solve our problem.


\section{Conclusions and Future Work}\label{sec:con_fut}

In this work, we propose a novel approach that enables privacy-preserving access to a remote quantum annealer for solving QUBO optimization problems. Our approach balances privacy, accuracy, and efficiency, allowing for trade-offs among these aspects by adjusting relevant parameters. Theoretical analysis demonstrates that our method effectively conceals the values of the elements in the model matrix, with the probability of an adversary accurately determining the ratio between them being negligible. Experimental evaluation shows that, with 5 matrices and a base of 4, our method achieves over 90\% accuracy within 300 samples.

Future research will focus on refining the weight vector selection process, developing more sophisticated methodologies for deriving final results, establishing theoretical bounds on accuracy, and testing our approach on a broader range of existing QUBO benchmarks for large-scale evaluation.

\section*{ACKNOWLEDGMENT}
The authors would like to thank all the reviewers for their
helpful comments. Qun Li was supported in part by the Commonwealth Cyber
Initiative (cyberinitiative.org).

\newpage

\section*{Appendix: Algorithm}\label{app}


\begin{algorithm}
\caption{Split and Permute the Model Matrix}\label{alg1}
\begin{algorithmic}

\Ensure \\
\begin{itemize}
    \item  $Q$ the model matrix
    \item  $k$ the number of matrices to be send
    \item  $r$ the base of rounding system
\end{itemize}    
\Require \\
\begin{itemize}
    \item  $\left\{M^\sigma_{\sigma'(1)}, M^\sigma_{\sigma'(2)}, ...,M^\sigma_{\sigma'(k)} \right\}$ the $k$ matrices to be sent to the server
    \item $\left\{ \sigma_1, \sigma_2, ... ,\sigma_k  \right\}$ the random permutations for the indices of the $k$ matrices respectively
    \item $\sigma'$ the random permutation for the indexes of the $k$ matrices
\end{itemize}    
\\
\For{ $i,j \in \left\{ 1, ... , n \right\}$ }
    \State  $ Q^*[i,j] \gets \frac{Q[i,j]}{|(1+\epsilon)max(Q)|}$, $\epsilon \rightarrow 0^+ $, $max(Q)$ is the element of $Q$ with the largest absolute value
\EndFor
\For{ $i,j \in \left\{1,2, ... ,n \right\},m \in \left\{1,2, ... ,k \right\}$  }
    \State  $ M_m[i,j]\gets \text{the}\ m \text{th} \ \text{digit} \ \text{of} \ Q^*[i,j] $ 
\EndFor 
\State Generate $ \sigma_1$, $\sigma_2$, ... , $\sigma_k $, $k$ random permutations of $\left\{1,2,...,n \right\}$ 
\For{$i,j \in \left\{1,2, ... ,n \right\},m \in \left\{1,2, ... ,k \right\}$}
    \State  $M_m^\sigma[i,j]:=M_m[\sigma_m(i),\sigma_m(j)]$
\EndFor
\State Generate $ \sigma'$, random permutation of $\left\{1,2,...,k \right\}$
\State Permute the indexes of the matrices $M_1^\sigma$, $M_2^\sigma$, ... , $M_k^\sigma$ using $\sigma'$
\\ \Return $\left\{M^\sigma_{\sigma'(1)}, M^\sigma_{\sigma'(2)}, ...,M^\sigma_{\sigma'(k)} \right\}$, $\left\{ \sigma_1, \sigma_2, ... ,\sigma_k  \right\}$, $\sigma'$ 

\end{algorithmic}
\end{algorithm}

\begin{algorithm}
\caption{Obtain the Final Result}\label{alg2}
\begin{algorithmic}
\Ensure \\
\begin{itemize}
    \item $ \left\{ v_1^{\sigma'} , v_2^{\sigma'} , ... , v_k^{\sigma'} \right\}$ the binary vectors send back by the server
    \item $t$ the number of samples
    \item $w$ the weight vector
    \item $\left\{ \sigma_1, \sigma_2, ... ,\sigma_k  \right\}$ the random permutations for the indices of the $k$ matrices respectively
    \item $\sigma'$ the random for the indexes of the $k$ matrices
\end{itemize}    
\Require \\
\begin{itemize}
    \item  $x$ the final result as a solution to the optimization problem
\end{itemize}    
\\
\For{ $m \in \left\{1,2, ...,k \right\}$ }
    \State  $v_m^\sigma \gets v_{\sigma^{'-1}(m)}^{\sigma'}$
\EndFor
\For{ $i \in \left\{1,2, ... ,n \right\},\  m\in \left\{1,2, ... ,k \right\}$  }
    \State  $v_m[i] \gets v_m^{\sigma}[\sigma^{-1}_m(i)]  $ 
\EndFor 
\State  $ \tilde{x} \gets \frac{\sum^{k}_{m=1}{w[m]\cdot v_m}}{\sum^k_{m=1}{w[m]}}$ 
\For{$i \in \left\{1,2, ... ,n \right\}$,\ $l \in \left\{1,2, ... ,t \right\}$}
    \State  $X_l \stackrel{\$}{\leftarrow}[0,1]^n,Pr[X_l[i]=1]=\tilde{x}[i]$
\EndFor
\\ 
\Return $\mathop{argmin}\limits_{x\in \left\{X_1,X_2,...,X_t \right\}} \ x^tQx$
\end{algorithmic}
\end{algorithm}

\begin{algorithm}
\caption{Main Algorithm}\label{alg3}
\begin{algorithmic}

\State \textbf{1.} (on the client) To solve QUBO for $Q$, run $Algo1(Q,k,r)$, sending $M^\sigma_{\sigma'(1)}$, $M^\sigma_{\sigma'(2)}$, ... , $M^\sigma_{\sigma'(k)}$ to the server.
\State \textbf{2.} (on the server) After receiving $M^\sigma_{\sigma'(1)}$, $M^\sigma_{\sigma'(2)}$, ... , $M^\sigma_{\sigma'(k)}$, solve the optimization problems for them respectively, and send the corresponding optimization solutions $  v_1^{\sigma'}$ , $v_2^{\sigma'}$ , ... , $v_k^{\sigma'} $ to the client.
\State \textbf{3.} (on the client) After receiving $  v_1^{\sigma'}$ , $v_2^{\sigma'}$ , ... , $v_k^{\sigma'} $, run $Algo2( \left\{ v_1^{\sigma'} , v_2^{\sigma'} , ... , v_k^{\sigma'} \right\},t,w,\left\{ \sigma_1, \sigma_2, ... ,\sigma_k  \right\},\sigma')$, obtaining the final result.

\end{algorithmic}
\end{algorithm}

\bibliographystyle{IEEEtran}
\bibliography{reference}

\begin{thebibliography}{10}
\providecommand{\url}[1]{#1}
\csname url@samestyle\endcsname
\providecommand{\newblock}{\relax}
\providecommand{\bibinfo}[2]{#2}
\providecommand{\BIBentrySTDinterwordspacing}{\spaceskip=0pt\relax}
\providecommand{\BIBentryALTinterwordstretchfactor}{4}
\providecommand{\BIBentryALTinterwordspacing}{\spaceskip=\fontdimen2\font plus
\BIBentryALTinterwordstretchfactor\fontdimen3\font minus \fontdimen4\font\relax}
\providecommand{\BIBforeignlanguage}[2]{{%
\expandafter\ifx\csname l@#1\endcsname\relax
\typeout{** WARNING: IEEEtran.bst: No hyphenation pattern has been}%
\typeout{** loaded for the language `#1'. Using the pattern for}%
\typeout{** the default language instead.}%
\else
\language=\csname l@#1\endcsname
\fi
#2}}
\providecommand{\BIBdecl}{\relax}
\BIBdecl

\bibitem{das2008colloquium}
A.~Das and B.~K. Chakrabarti, ``Quantum annealing and analog quantum computation,'' \emph{Reviews of Modern Physics}, vol.~80, no.~3, p. 1061, 2008.

\bibitem{battaglia2005optimization}
D.~A. Battaglia, G.~E. Santoro, and E.~Tosatti, ``Optimization by quantum annealing: Lessons from hard satisfiability problems,'' \emph{Physical Review E}, vol.~71, no.~6, p. 066707, 2005.

\bibitem{rajak2023quantum}
A.~Rajak, S.~Suzuki, A.~Dutta, and B.~K. Chakrabarti, ``Quantum annealing: An overview,'' \emph{Philosophical Transactions of the Royal Society A}, vol. 381, no. 2241, p. 20210417, 2023.

\bibitem{kochenberger2014unconstrained}
G.~Kochenberger, J.-K. Hao, F.~Glover, M.~Lewis, Z.~L{\"u}, H.~Wang, and Y.~Wang, ``The unconstrained binary quadratic programming problem: a survey,'' \emph{Journal of combinatorial optimization}, vol.~28, pp. 58--81, 2014.

\bibitem{glover2019quantum}
F.~Glover, G.~Kochenberger, and Y.~Du, ``Quantum bridge analytics i: a tutorial on formulating and using qubo models,'' \emph{4or}, vol.~17, no.~4, pp. 335--371, 2019.

\bibitem{mucke2019learning}
S.~M{\"u}cke, N.~Piatkowski, and K.~Morik, ``Learning bit by bit: Extracting the essence of machine learning.'' in \emph{LWDA}, 2019, pp. 144--155.

\bibitem{date2021qubo}
P.~Date, D.~Arthur, and L.~Pusey-Nazzaro, ``Qubo formulations for training machine learning models,'' \emph{Scientific reports}, vol.~11, no.~1, p. 10029, 2021.

\bibitem{websiteCPLEX}
\BIBentryALTinterwordspacing
$\text{IBM ILOG}$. $\text{ILOG CPLEX Optimization Studio Python tutorial}$. [Online]. Available: \url{https://www.ibm.com/docs/en/icos/22.1.1?topic=tutorials-python-tutorial}
\BIBentrySTDinterwordspacing

\bibitem{wu2024distributed}
J.~Wu, T.~Hu, and Q.~Li, ``Distributed quantum machine learning: Federated and model-parallel approaches,'' \emph{IEEE Internet Computing}, vol.~28, no.~2, pp. 65--72, 2024.

\bibitem{wu2022scalable}
J.~Wu, Z.~Tao, and Q.~Li, ``Scalable quantum neural networks for classification,'' in \emph{2022 IEEE International Conference on Quantum Computing and Engineering (QCE)}.\hskip 1em plus 0.5em minus 0.4em\relax IEEE, 2022, pp. 38--48.

\bibitem{wu2023more}
J.~Wu, T.~Hu, and Q.~Li, ``More: Measurement and correlation based variational quantum circuit for multi-classification,'' in \emph{2023 IEEE International Conference on Quantum Computing and Engineering (QCE)}, vol.~1.\hskip 1em plus 0.5em minus 0.4em\relax IEEE, 2023, pp. 208--218.

\bibitem{huang2015differentially}
Z.~Huang, S.~Mitra, and N.~Vaidya, ``Differentially private distributed optimization,'' in \emph{Proceedings of the 16th International Conference on Distributed Computing and Networking}, 2015, pp. 1--10.

\bibitem{han2016differentially}
S.~Han, U.~Topcu, and G.~J. Pappas, ``Differentially private distributed constrained optimization,'' \emph{IEEE Transactions on Automatic Control}, vol.~62, no.~1, pp. 50--64, 2016.

\bibitem{shoukry2016privacy}
Y.~Shoukry, K.~Gatsis, A.~Alanwar, G.~J. Pappas, S.~A. Seshia, M.~Srivastava, and P.~Tabuada, ``Privacy-aware quadratic optimization using partially homomorphic encryption,'' in \emph{2016 IEEE 55th Conference on Decision and Control (CDC)}.\hskip 1em plus 0.5em minus 0.4em\relax IEEE, 2016, pp. 5053--5058.

\bibitem{8262869}
A.~B. Alexandru, K.~Gatsis, and G.~J. Pappas, ``Privacy preserving cloud-based quadratic optimization,'' in \emph{2017 55th Annual Allerton Conference on Communication, Control, and Computing (Allerton)}, 2017, pp. 1168--1175.

\bibitem{7347330}
L.~Zhou and C.~Li, ``Outsourcing large-scale quadratic programming to a public cloud,'' \emph{IEEE Access}, vol.~3, pp. 2581--2589, 2015.

\bibitem{8430960}
S.~Gade and N.~H. Vaidya, ``Private optimization on networks,'' in \emph{2018 Annual American Control Conference (ACC)}, 2018, pp. 1402--1409.

\bibitem{zhang2022quadratic}
Z.~Zhang, X.~Che, X.~Jiao, W.~Yu, and L.~Wan, ``Quadratic optimization using additive homomorphic encryption in cps,'' in \emph{2022 13th Asian Control Conference (ASCC)}.\hskip 1em plus 0.5em minus 0.4em\relax IEEE, 2022, pp. 1995--2000.

\bibitem{han2009privacy}
S.~Han, W.~K. Ng, L.~Wan, and V.~C. Lee, ``Privacy-preserving gradient-descent methods,'' \emph{IEEE transactions on knowledge and data engineering}, vol.~22, no.~6, pp. 884--899, 2009.

\bibitem{toft2009solving}
T.~Toft, ``Solving linear programs using multiparty computation,'' in \emph{International conference on financial cryptography and data security}.\hskip 1em plus 0.5em minus 0.4em\relax Springer, 2009, pp. 90--107.

\bibitem{Salinas2016}
\BIBentryALTinterwordspacing
S.~Salinas, C.~Luo, W.~Liao, and P.~Li, ``Efficient secure outsourcing of large-scale quadratic programs,'' in \emph{Proceedings of the 11th ACM on Asia Conference on Computer and Communications Security}, ser. ASIA CCS '16.\hskip 1em plus 0.5em minus 0.4em\relax New York, NY, USA: Association for Computing Machinery, 2016, p. 281–292. [Online]. Available: \url{https://doi.org/10.1145/2897845.2897862}
\BIBentrySTDinterwordspacing

\bibitem{dwork2006differential}
C.~Dwork, ``Differential privacy,'' in \emph{International colloquium on automata, languages, and programming}.\hskip 1em plus 0.5em minus 0.4em\relax Springer, 2006, pp. 1--12.

\end{thebibliography}


\end{document}